\documentclass[aap]{imsart}

\RequirePackage{amsthm,amsmath,amsfonts,amssymb}
\RequirePackage[numbers]{natbib}
\RequirePackage[colorlinks,citecolor=blue,urlcolor=blue]{hyperref}
\RequirePackage{graphicx}

\startlocaldefs
\theoremstyle{plain}

\newtheorem{theorem}{Theorem}[section]
\newtheorem{lemma}[theorem]{Lemma}
\newtheorem{cor}[theorem]{Corollary}
\newtheorem{prop}[theorem]{Proposition}
\DeclareMathOperator{\diag}{diag}

\theoremstyle{remark}
\newtheorem{definition}[theorem]{Definition}

\endlocaldefs

\begin{document}

\begin{frontmatter}
\title{An Upper Bound of the Information Flow\\ From Children to Parent Node on Trees}
%\title{A sample article title with some additional note\thanksref{t1}}
\runtitle{A Bound of Information Flow on Trees}
%\thankstext{T1}{A sample additional note to the title.}

\begin{aug}
%%%%%%%%%%%%%%%%%%%%%%%%%%%%%%%%%%%%%%%%%%%%%%%
%% Only one address is permitted per author. %%
%% Only division, organization and e-mail is %%
%% included in the address.                  %%
%% Additional information can be included in %%
\author[A]{\fnms{Cassius} \snm{Manuel}\ead[label=e1]{cassius.perez@univie.ac.at}}
%%%%%%%%%%%%%%%%%%%%%%%%%%%%%%%%%%%%%%%%%%%%%%
%% Addresses                                %%
%%%%%%%%%%%%%%%%%%%%%%%%%%%%%%%%%%%%%%%%%%%%%%
\address[A]{Center for Integrative Bioinformatics Vienna, Max Perutz Labs
(University of Vienna and Medical University of Vienna),
\printead{e1}}
\end{aug}

\begin{abstract}
We consider the transmission of a state from the root of a tree towards its leaves, assuming that each transmission occurs through a noisy channel. The states at the leaves are observed, while at deeper nodes we can compute the likelihood of each state given the observation. In this sense, information flows from child nodes towards the parent node.

Here we find an upper bound of this children-to-parent information flow. To do so, first we introduce a new measure of information, the memory vector, whose norm quantifies whether all states have the same likelihood. Then we find conditions such that the norm of the memory vector at the parent node can be linearly bounded by the sum of norms at the child nodes. 

We also describe the reconstruction problem of estimating the ancestral state at the root given the observation at the leaves.  We infer sufficient conditions under which the original state at the root cannot be confidently reconstructed using the observed leaves, assuming that the number of levels from the root to the leaves is large. 
 
\end{abstract}



\end{frontmatter}
%%%%%%%%%%%%%%%%%%%%%%%%%%%%%%%%%%%%%%%%%%%%%%
%% Please use \tableofcontents for articles %%
%% with 50 pages and more                   %%
%%%%%%%%%%%%%%%%%%%%%%%%%%%%%%%%%%%%%%%%%%%%%%
%\tableofcontents

\section{Introduction}

Consider the following broadcasting process. At the root $R$ of a tree, sample a state of alphabet ${\mathbb{A} = \{ 0, 1, \cdots, K\}}$ following distribution $\bm{\mu}>0$. The tree is $d$-ary (every node has at most $d$ children) and has $g$ levels (the unique path from the root $R$ to each of the leaves has $g$ edges).
The sampled state at $R$ is then transmitted independently to each child node of $R$ through a noisy channel, described by an aperiodic, irreducible Markov matrix $P$. Thus in channel $P= (p_{ij})$, where $i,j \in \mathbb{A}$, entry $p_{ij}$ is the probability of state $i$ mutating to state $j$ during transmission.

\medskip

The transmission of states to the next level continues analogously until a state is transmitted to every leaf of the tree. The ordered set of states at the leaves is called a \textit{pattern}, represented by symbol $\partial$. See Figure \ref{fig:tree_example} for an example of this process with $\mathbb{A} = \{ 0,1,2,3\}$.
\begin{figure}[tb]
    \centering
    \includegraphics{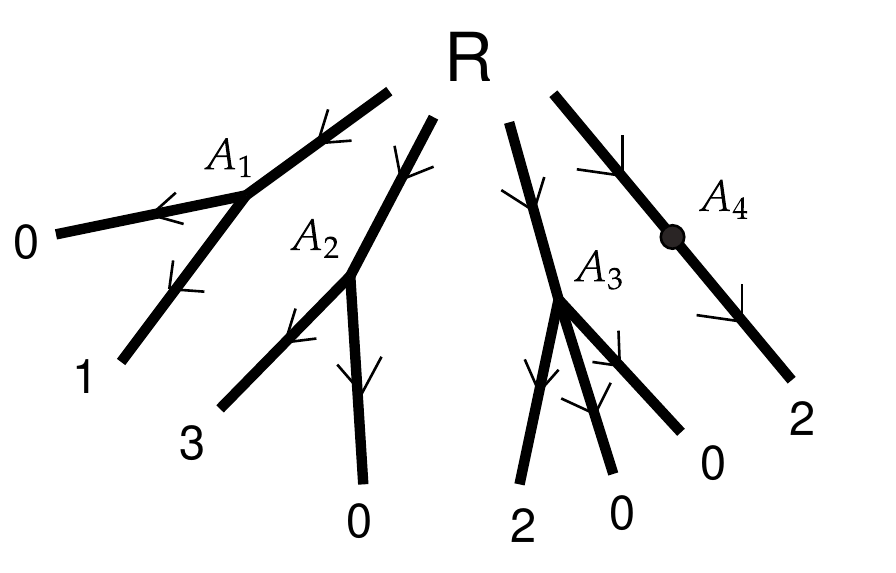}
    \caption{Example of broadcasting process on a $2$-level, incomplete $4$-ary tree. The ancestral state $i \in \mathbb{A}$ at the root $R$ is sent independently to each node $A_1, \cdots, A_4$. The probability of $i$ mutating to $k$ during this transmission is $p_{ik}$. Then, the state $i_c$ at each node $A_c$ is sent independently to each of its child leaves, with a probability $p_{i_ck}$ of mutating to $k$ during transmission. The observed pattern is $\partial = 01302002$, while the four subpatterns determined by clades $A_1, \cdots , A_4$ are $\partial_1 = 01$, $\partial_2 = 30$, $\partial_3 = 200$ and $\partial_4 = 2$.}
    \label{fig:tree_example}
\end{figure}

\medskip

We address the following problem: Estimate the ancestral state at the root $R$ from the pattern $\partial$ assuming that $g$ is large,  if the tree, the prior $\bm{\mu}$ and channel $P$ are known. Notably, Maximum Likelihood (ML) has the highest probability of a correct reconstruction among all reconstruction methods (see Theorem 17.2, \cite{guiasu1977}). ML computes the posterior state distribution $\bm{r_\partial}$ at the root given $\partial$. The ML estimate is the state with the maximum posterior probability, also called the Maximum A Posteriori (MAP) estimate. 

\medskip

Here we find sufficient conditions making the reconstruction problem unsolvable, meaning that $\bm{r_\partial} = \bm{\mu}$ for all patterns $\partial$ assuming that $g \to \infty$. More explicitly, in Theorem \ref{prop:unsolvable} we consider a $d$-ary tree and the same channel $P$ for all transmissions, with equilibrium frequency $\bm{\pi}$. We prove that, for some constant $C>0$ depending on channel $P$, if
\[C d < \min \Big\{ \frac{1}{3}, \frac{\min_i \pi_i}{\sqrt{1-\min_i \pi_i}} \Big\},\]
then the reconstruction problem is unsolvable for every pattern $\partial$. As an example, in Corollary \ref{cor:boundREVERSIBLE} we show that $C = |\theta_1|$ when the channel $P$ is reversible, where $\theta_1 \in (-1, 1)$ is the non-unitary absolutely largest eigenvalue of channel $P$.

\medskip

If the reconstruction problem is unsolvable, then the MAP estimate is identical to the maximum \textit{a priori} estimate, that is, the ML estimate using only prior $\bm{\mu}$. Since ML has an optimal probability of a correct reconstruction, we can say that, under these conditions, any pattern $\partial$ is uninformative about the ancestral root state.

\medskip
This paper is structured as follows. In Section \ref{sec:partial_likelihoods}, we describe the likelihood computation on a tree, and then in Section \ref{sec:generalProperties} we introduce a new measure of the information carried by a likelihood vector, namely the $L_2(\pi)$-norm of the memory vector. We rewrite the likelihood vector using the memory vector in Section \ref{sec:rewriting_likelihood}, which is useful for the bounds of Section \ref{sec:try_again}. 

Our core result is Theorem \ref{prop:rootByChildren}, where we bound the memory-vector norm at a parent node by the sum of memory-vector norms at the child nodes. This theorem is very versatile, since it can be applied to any tree structure, even with a different channel for every child node. The possible definitions of unsolvability are explained in Section \ref{sec:definitions}. Using the bound of Theorem \ref{prop:rootByChildren}, in Section \ref{sec:bounds} we give sufficient conditions on matrix $P$ such that the reconstruction problem is unsolvable. Unsolvability in expectation is studied in Section \ref{appn}.

\section{Comparison with previous results} The solvability of the reconstruction problem has been studied repeatedly in physics, statistic and phylogenetics (see \cite{Mossel2001SurveyIF}). Notably, focus has been placed on solvability \textit{in expectation}, meaning the following: If pattern $\partial$ has a probability $\Pr(\partial)$ of being observed, then characterize whether $\mathbb{E}[\| \bm{r_\partial} - \bm{\mu} \|] \to 0$ as $g \to \infty$.  See Section \ref{sec:definitions} regarding the possible definitions of unsolvability.

\medskip

Assuming a binary symmetric channel (the so called Ising model), in \cite{bleher1995} it was proven that $d|\theta_1|^2<1$ implies that the reconstruction problem is unsolvable in expectation, where $d$ is the $d$-arity of the tree and $\theta_1 \in (-1, 1)$ is the non-unitary absolutely largest eigenvalue of matrix $P$. Two generalizations of the Ising model were studied later, namely the binary asymmetric channel and the symmetric $\mathbb{A}$-dimensional channel (Potts model in physics \cite{Mossel2001SurveyIF}; Jukes-Cantor model in phylogenetics \cite{librofilogenia}). In \cite{BeatingtheSecondEigenvalue} and \cite{Informationflowontrees}, it is stated that \textit{for those two kind of channels}, the reconstruction problem is unsolvable in expectation if ${d|\theta_1| < 1}$. Solvable examples in expectation when ${d|\theta_1| > 1}$ are presented in \cite{BeatingtheSecondEigenvalue}.

%Since ML has an optimal reconstruction probability, we ignore the simpler census method, which received much attention in \cite{evans2000},  \cite{BeatingtheSecondEigenvalue} and \cite{Informationflowontrees}. 

\medskip

In contrast, our results apply to any noisy channel, as long as it is irreducible and aperiodic, which are typical assumptions. Notably, we do not rely on the entropy-based theory of communication described by \cite{shannon1948}, as \cite{BeatingtheSecondEigenvalue} and \cite{Informationflowontrees} do. Instead, in Subsection \ref{subsec:memory_vector} we propose a new measure of the information carried by a likelihood vector, namely the $L_2(\pi)$-norm of the memory vector.

\section{Likelihoods on a d-ary Tree} \label{sec:partial_likelihoods}
The Pruning algorithm to recursively compute the likelihood of the model parameters given a pattern was introduced in \cite{Felsenstein1981}.
In this section, we rephrase the Pruning algorithm for a $d$-ary tree with a different channel for each child node.
    
We want to compute the likelihood of each state at the root $R$ given pattern $\partial$. To that end, we define the likelihood vector at the root as \begin{equation}\bm{\rho_\partial} := (\text{Pr}(\partial\;|\; R = 0), \cdots , \text{Pr}(\partial \; |\; R = K)).\end{equation} 
Figure \ref{fig:multitree} shows the root $R$ and its $d$ outgoing channels $P_1, \cdots, P_d$  with their respective child nodes $A_1, \cdots, A_d$. Abusing of the notation, we define clade $A_c$ as the subtree rooted at node $A_c$ containing all descendant nodes from $A_c$. The leaves of clade $A_c$ determine subpattern $\partial_c$. Considered independently, each clade allows to compute a likelihood vector $\bm{\alpha_c}$ at the root $A_c$ given subpattern $\partial_c$. More explicitly, we define
\begin{equation}\aaa := (\text{Pr}(\partial_c \;|\; A_c = 0), \cdots , \text{Pr}(\partial_c \; |\; A_c = K)).\end{equation}

 \begin{figure}[tb]
    \includegraphics[width=0.6\linewidth]{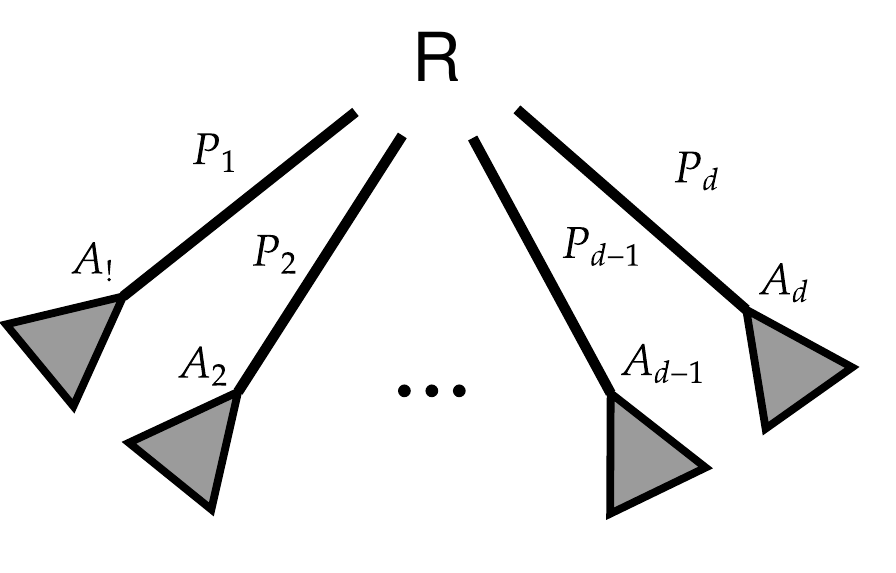}
	\caption{The state at node $A_c$, for $c \in [d]$, was obtained by transmitting the state at the root $R$ through a channel $P_c$. The likelihood vector $\bm{\rho_\partial}$ at $R$  satisfies $\bm{\rho_\partial} = \bigcomp_{c \in [d]} P_c\aaa$.} \label{fig:multitree}
\end{figure}

The Pruning algorithm consists in expressing  $\bm{\rho_\partial}$ in terms of the likelihood vectors $\bm{\alpha_1}, \cdots , \bm{\alpha_d}$.  Using the law of total probability, the likelihood vector at the root $R$ considering \textit{only} clade $A_c$ is $P_c\aaa$. More explicitly, 
\begin{equation}
    P_c\aaa = (\text{Pr}(\partial_c \;|\; R = 0), \cdots , \text{Pr}(\partial_c \; |\; R = K)).
\end{equation}

Now we use the fact that the state at each node $A_c$ was obtained \textit{independently} from the state at $R$. This implies that $\rrr$ equals the entrywise product of all likelihood vectors $P_c\aaa$. To visualize this better, if we focus on state $R = 0$, it holds that 

\begin{equation}
    \text{Pr}(\partial \;|\; R = 0) =  \prod _{c \in [d]} \text{Pr}(\partial_c \;|\; R = 0).
\end{equation}

\begin{sloppypar}
 We denote the entrywise product between two vectors $\bm{v} = (v_i)$ and $\bm{w} = (w_i)$ as ${\bm{v} \circ \bm{w} := (v_i w_i)}$. Given vectors $\bm{v_1}, \cdots , \bm{v_d}$, their entrywise product is denoted as ${\bigcomp_{c \in [d]} \bm{v_c} := \bm{v_1} \circ \cdots \circ \bm{v_d}.}$ With this notation, we can write
\end{sloppypar}

\begin{equation} \label{eq:pruning}
    \rrr = \bigcomp_{c \in [d]} P_c \aaa.
\end{equation}
Now we just have to proceed analogously: Every vector $\aaa$ can be expressed in terms of the child nodes of $A_c$, and so on, until we eventually arrive to the likelihood vectors at the leaves. The likelihood vector at a leaf where state $i$ was observed is the canonical vector $\bm{e_i}$, composed by $0$'s except for a $1$ at coordinate $i$. Thus the Pruning algorithm is a recursion based on Eq. \ref{eq:pruning} where the leaves have known likelihood vectors and the computation proceeds towards the root of the tree.

Often one needs to compute $\text{Pr}(\partial)$, that is, the probability of observing pattern $\partial$. This computation requires the assumption of some prior state distribution $\bm{\mu} > 0$ at node $R$. The law of total probability implies that \begin{equation}
    \text{Pr}(\partial) = \bm{\mu}\cdot \bm{\rho_\partial},
\end{equation}
where $"\cdot"$ denotes the Euclidean dot product. 

%For the results of this work, it is actually unnecessary to assume a prior, although Section \ref{sec:rewriting_likelihood} has a probabilistic interpretation that assumes a stationary process, that is, prior $\bm{\mu} = \bm{\pi} > 0$.

\section{The norm of the memory vector}
\label{sec:generalProperties}

\subsection{The equilibrium frequency}
  The equilibrium frequency of the noisy channel plays a fundamental role in this work. The necessary preliminaries about the equilibrium frequency can be summarized as follows (see  \cite{levin2017markov}, Chapter $1$ and Section $12.1$).
 
\begin{prop} \label{prop:reversible_posi}
For an irreducible and aperiodic Markov matrix $P$, the following holds:
\begin{itemize}
    \item[a)] Matrix $P$ has a unique equilibrium distribution $\bm{\pi} > 0$, defined by equation ${\bm{\pi}^T  P = \bm{\pi}^T}$.
    \item[b)]  Matrix $P$ has eigenvalue $\theta_0 = 1$ with algebraic multiplicity $1$, and the other eigenvalues $\{ \theta_1, \cdots, \theta_K \} \subset \mathbb{C}$ have a module smaller than $1$. We assume ${1 > |\theta_1| \geq \cdots \geq |\theta_K|}.$
    \item[c)] If $P$ is reversible, all eigenvalues are real.  Moreover, matrix $P$ has an orthogonal basis of right eigenvectors $\bm{v_k}$ and left eigenvectors $\bm{h_k}^T$ such that $\bm{h_k} = \bm{\pi} \circ \bm{v_k}$ for $k \in [0,K]$,  where $\bm{v_0} = \bm{1}$,  $\bm{h_0} = \bm{\pi}$ and
\[P =  \bm{1}\bm{\pi}^T + \bm{v_1} \bm{h_1}^T \theta_1 + \cdots + \bm{v_K} \bm{h_K}^T \theta_K. \]
\end{itemize}
\end{prop}

\subsection{General properties of the norm}

For some definitions of this subsection, we follow the notation of \cite{levin2017markov}, Section $12.5$. Given distribution $\bm{\pi} > 0$, for any two vectors $\bm{x} = (x_i)$ and $\bm{y} = (y_i)$ where $i\in\mathbb{A}$, we define the $\pi$-inner product as
\begin{equation}
     \langle \bm{x}, \bm{y}\rangle_\pi :=
\bm{\pi} \cdot (\bm{x} \circ \bm{y}) = \sum_i \pi_i x_i y_i.
\end{equation}

The $\pi$-inner product induces the $L_2(\pi)$-norm, defined as
 \begin{equation}
     \|\bm{x}\|_\pi :=  \sqrt{\langle \bm{x}, \bm{x}\rangle_\pi}. 
 \end{equation}
 As in every inner product space, the 
 %Cauchy-Schwartz inequality 
%\begin{equation}
%    \langle \bm{x}, \bm{y} \rangle_\pi \leq \|\bm{x} \|_\pi \| \bm{y}\|_\pi,
%\end{equation}
%and also the 
triangle inequality holds. Explicitly,
\begin{equation} \label{eq:chainNorms}
    \|\bm{x} + \bm{y} \|_\pi \leq \|\bm{x}\|_\pi + \| \bm{y}\|_\pi.
\end{equation}
 \begin{sloppypar}
  Denote the Euclidean norm of $\bm{x}$ as $\|\bm{x}\| := \sqrt{\sum_i x^2_i}$, and the uniform norm as ${\|\bm{x}\|_\infty := \max_i |x_i|}$. Since $|x_i| \leq\|\bm{x}\|$ for all $i$ and $\sqrt{\min_i} \pi_i \| \bm{x}\|\leq \|\bm{x}\|_\pi$, it holds that
\begin{equation} \label{eq:entryBYnorm}
    \|\bm{x}\|_\infty \leq \|\bm{x}\| \leq  \frac{\|\bm{x}\|_\pi}{\sqrt{\min_i \pi_i}}.
\end{equation}
 \end{sloppypar}

An important reason to use the $L_2(\pi)$-norm is that a vector cannot increase its $L_2(\pi)$-norm when we centralize its entries with weights $\bm{\pi}$, as described in the following proposition.
\begin{prop}[Centralizing inequality] \label{prop:inequality_L2}
Given vector $\bm{x}$, the $L_2(\pi)$-norm satisfies
\[ \| \bm{x} - \bm{1} \langle \bm{x}, \bm{1} \rangle_\pi \|^2_\pi = \| \bm{x}\|^2_\pi  - \langle \bm{x}, \bm{1}\rangle^2_\pi.\]
In particular, it holds that
\[ \| \bm{x} - \bm{1} \langle \bm{x}, \bm{1} \rangle_\pi \|_\pi \leq \|\bm{x}\|_\pi.\]
\end{prop}

\begin{proof} Developing the $\pi
$-inner product, we obtain \begin{align} 
\Big\langle \bm{x} - \bm{1} \langle \bm{x}, \bm{1} \rangle_\pi, \bm{x} - \bm{1} \langle \bm{x}, \bm{1} \rangle_\pi \Big \rangle_\pi = \langle \bm{x}, \bm{x}\rangle_\pi -  2\langle \bm{x}, \bm{1} \rangle^2_\pi + \langle \bm{x}, \bm{1}\rangle^2_\pi \langle \bm{1}, \bm{1} \rangle_\pi.
\end{align}

Since $\langle \bm{1}, \bm{1} \rangle_\pi = 1$, it follows that
\[ \| \bm{x} - \bm{1} \langle \bm{x}, \bm{1} \rangle_\pi \|^2_\pi = \| \bm{x}\|^2_\pi  - \langle \bm{x}, \bm{1}\rangle^2_\pi,\]
as desired. To obtain the inequality, use the fact that $\langle \bm{x}, \bm{1}\rangle^2_\pi \geq 0$, giving
\[ \| \bm{x} - \bm{1} \langle \bm{x}, \bm{1} \rangle_\pi \|^2_\pi \leq \| \bm{x}\|^2_\pi.\]
\end{proof}

\subsection{The memory vector} \label{subsec:memory_vector}
Consider the likelihood vector $\rrr$ at node $R$ determined by pattern $\partial$. We define the normalized likelihood vector $\rrt := \rrr/(\rrr \cdot \bm{\pi})$ and the distribution $\bm{r_\pi} := \rrt \circ \bm{\pi}$. With analogous definitions, given the likelihood vector $\bm{\alpha}$ at $A$ and pattern $\pred$, $\bm{\tilde{\alpha}} := \bm{\alpha}/(\bm{\alpha} \cdot \bm{\pi})$ and $\bm{a_\pi} := \bm{\tilde{\alpha}} \circ \bm{\pi}$. Note that, assuming prior $\bm{\mu} = \bm{\pi}$ (so called stationarity), it holds that $\text{Pr}(\partial) = \rrr \cdot \bm{\pi}$, $\text{Pr}(\pred) = \bm{\alpha} \cdot \bm{\pi}$, and thus $\bm{r_\pi}$ and $\bm{a_\pi}$ are the posteriors under this assumption.

The $\pi$-inner product of $\rrt$ and $\bm{\tilde{\alpha}}$ is
\begin{equation}
     \langle \rrt, \bm{\tilde{\alpha}}\rangle_\pi :=
\bm{\pi} \cdot (\rrt \circ \bm{\tilde{\alpha}}) = \bm{r_\pi} \cdot \bm{\tilde{\alpha}} = \bm{a_\pi} \cdot  \rrt.
\end{equation}

\begin{sloppypar}
 Notably, we have $\langle \rrt, \bm{1}\rangle_\pi  = \rrt \cdot \bm{\pi} = 1$. Using the Cauchy-Schwartz inequality, if we write ${\rrt = (\tilde{\rho}_0, \cdots, \tilde{\rho}_K)}$, we obtain the lower bound
\end{sloppypar} \begin{equation}
   \|\rrt\|^2_\pi :=  \langle \rrt, \rrt\rangle_\pi = \sum_i \pi_i \tilde{\rho}^2_i= (\sum_i \pi_i) (\sum_i \pi_i \tilde{\rho}^2_i)  \geq  \sum_i \sqrt{\pi_i} \sqrt{\pi_i} \tilde{\rho}_i = 1.
\end{equation}
%\begin{equation}
%    \langle \bm{1}, \bm{1}\rangle_\pi  = \bm{1} \cdot \bm{\pi} = 1.
%\end{equation}

A useful upper bound when dealing with normalized likelihood vectors is 
\begin{equation}
    \|\rrt\|^2_\pi = \bm{r_\pi}^T \rrt \leq \max_k \tilde{\rho}_k \leq \frac{1}{\min_i \pi_i},
\end{equation} 
where we used the fact that $\rrt\cdot \bm{\pi} = 1$, and therefore for all $k$, $\tilde{\rho}_k \leq 1/\pi_k \leq  1/\min_i \pi_i$. 

However, we are more interested in the distance from $\rrt$ to $\bm{1}$. Intuitively, as $\rrt$ approaches $\bm{1}$, all ancestral states become equally likely, that is, pattern $\partial$ becomes less informative. This motivates the following definition.

\begin{definition}
We say that vector $\rrt-\bm{1}$ is  \textit{the memory vector at node $R$ given $\partial$}.
\end{definition}
When needed, the memory vector will be denoted as $\bm{m}$, with adequate indices depending on the context. The memory vector satisfies
\begin{equation} \label{eq:memory}
    0 \leq  \|\rrt-\bm{1}\|^2_\pi = \langle \rrt-\bm{1}, \rrt-\bm{1}\rangle_\pi = \|\rrt\|^2_\pi - 1  \leq \frac{1}{\min_i \pi_i} - 1,
\end{equation} 
that is, the $L_2(\pi)$-norm of any memory vector is upper bounded by $\sqrt{1/\min_i \pi_i - 1}$.

\subsection{Endomorphism of memory vectors}
\begin{sloppypar}
 Consider the process of transmission using channel $P$ with  equilibrium frequency $\bm{\pi}$. If node $A$ received its state from a node, say $PA$, then given $\partial_A$, node $PA$ has likelihood vector $P\bm{\alpha}$. Importantly, since ${\bm{\pi}^TP\bm{\alpha} = \bm{\pi}^T\bm{\alpha} = \bm{\pi} \cdot \bm{\alpha}}$, it follows that the normalized likelihood vector at node $PA$ is $P \bm{\tilde{\alpha}}$. Thus matrix $P$ is an endomorphism of normalized likelihood vectors. Consequently, the memory vector at node $PA$ is $P\bm{\tilde{\alpha}} - \bm{1}$. Moreover, since $P\bm{1} = \bm{1}$, it holds that
 \[P\bm{\tilde{\alpha}} - \bm{1} = P(\bm{\tilde{\alpha}} - \bm{1}),\]
and thus matrix $P$ is also an endomorphism of memory vectors. In Section \ref{sec:bounds}, we will assume that $\bm{\pi}$ is the equilibrium frequency of all channels to preserve this endomorphism.
\end{sloppypar}

\subsection{The reversible case}

If channel $P$ is reversible, then the eigendecomposition of Prop. \ref{prop:reversible_posi}.c applies, yielding simple ways to describe the $L_2(\pi)$-norm and the action of $P$. Indeed, using decomposition $I = \bm{1} \bm{\pi}^T  + \sum_{k\in  [K]}  \bm{v_k} \bm{h_k}^T$  where $\bm{h_k} = \bm{\pi}\circ \bm{v_k}$, we can rewrite the $L_2(\pi)$-norm as 
\begin{equation} \label{eq:decomposition_memory}
     \|\rrt\|^2_\pi = \bm{r_\pi}^T \rrt = \bm{r_\pi}^T  I \rrt = 1 +  \sum_{k \in [K]} (\rrt\cdot \bm{h_k})^2 \geq 1.
\end{equation}
 Regarding the action of a reversible Markov matrix $P$ on a normalized likelihood vector $\bm{\tilde{\alpha}}$, it can be alternatively described as
\begin{equation}
    P\bm{\tilde{\alpha}} = \bm{1} + \sum_{k \in  [K]} \theta_k (\bm{\tilde{\alpha}}\cdot \bm{h_k}) \bm{v_k},
\end{equation}
where we used the eigendecomposition of  Prop. \ref{prop:reversible_posi}.c, $P = \bm{1} \bm{\pi}^T  + \sum_{k\in  [K]}  \theta_k \bm{v_k} \bm{h_k}^T$. Recall that the eigenvalues $\theta_k$ satisfy $1 > |\theta_1| \geq \cdots \geq |\theta_K|$. Using Equations \ref{eq:memory} and \ref{eq:decomposition_memory} with ${\rrt = P\bm{\tilde{\alpha}}}$, the memory vector at node $PA$ satisfies
\begin{equation}
     \|P\bm{\tilde{\alpha}}-\bm{1}\|^2_\pi =  \sum_{k \in  [K]} \theta_k^2 (\bm{\tilde{\alpha}}\cdot \bm{h_k})^2.
\end{equation}
Interestingly, this yields the bound
\begin{equation} \label{eq:bound_memoryvector1}
      \|P\bm{\tilde{\alpha}}-\bm{1}\|^2_\pi \leq  \theta^2_1 \sum_{k \in  [K]} (\bm{\tilde{\alpha}}\cdot \bm{h_k})^2 = \theta_1^2 \; \|\bm{\tilde{\alpha}}-\bm{1}\|^2_\pi,
\end{equation} 
or taking the square root,
\begin{equation} \label{eq:bound_memoryvector}
      \|P\bm{\tilde{\alpha}}-\bm{1}\|_\pi \leq  |\theta_1| \; \|\bm{\tilde{\alpha}}-\bm{1}\|_\pi 
      %\leq |\theta_1| \; \sqrt{\frac{1}{\min_i \pi_i} - 1}
      .
\end{equation} 
%where for last inequality we used Eq. \ref{eq:memory}. 
Hence the $L_2(\pi)$-norm of the memory vector decreases at least by a factor of $|\theta_1|$ under the action of the reversible matrix $P$.

\subsection{A general bound of the norm growth}

The bounds of norm growth as the one of Eq. \ref{eq:bound_memoryvector} are useful for our proofs. Without the assumption that matrix $P$ be reversible, a general bound similar to Eq. \ref{eq:bound_memoryvector} as
\begin{equation} \label{eq:bound_memoryvectorGENERAL}
      \|P\bm{\tilde{\alpha}}-\bm{1}\|_\pi \leq  C \; \|\bm{\tilde{\alpha}}-\bm{1}\|_\pi 
      %\leq |\theta_1| \; \sqrt{\frac{1}{\min_i \pi_i} - 1}
\end{equation} 
for some constant $C>0$ is more difficult to obtain, since matrix $P$ may have complex eigenvalues or not diagonalize. A weak bound using well known properties of the singular value decomposition (explained in  \cite{horn_johnson_1991}) can be obtained as follows.

Let us define $\Pi := \bm{1}\bm{\pi}^T$ and $\Psi := \diag(\bm{\pi})$. Since $P\bm{1} =  \Pi \bm{1} = \Pi \bm{\tilde{\alpha}} = \bm{1}$ for any normalized likelihood vector $\bm{\tilde{\alpha}}$, it holds that \begin{equation}
    P\bm{\tilde{\alpha}}-\bm{1} = (P-\Pi)  (\bm{\tilde{\alpha}}-\bm{1}).
\end{equation}
Notably, matrix $P -\Pi$ has eigenvalues $0, \theta_1, \cdots, \theta_K$, although we only need its singular values. Consider the largest singular value $\sigma_1$ of $P -\Pi$, which satisfies $\sigma_1 \geq |\theta_1|$ and thus may be larger than $1$ (see \cite{horn_johnson_1991}). Let $\| \cdot \|$ denote the matrix norm induced by the Euclidean norm.  It holds that
\begin{align}
    \| P\bm{\tilde{\alpha}}-\bm{1} \|_\pi &= \|(P-\Pi)  (\bm{\tilde{\alpha}}-\bm{1})\|_\pi = \nonumber\\ 
    &= \| \sqrt{\Psi} (P-\Pi)  (\bm{\tilde{\alpha}}-\bm{1}) \| \leq \| \sqrt{\Psi} \| \;  \| P-\Pi\| \;  \|\bm{\tilde{\alpha}}-\bm{1} \|,
\end{align}
\begin{sloppypar}
 where the inequality follows because induced norms are sub-multiplicative. Since the Euclidean matrix norm coincides with the spectral norm, ${\| \sqrt{\Psi} \| = \max_i \sqrt{\pi_i}}$ and ${\| P - \Pi\| = \sigma_1}$. This together with $\| \bm{x}\| \leq \| \bm{x}\|_\pi/\min_i \sqrt{\pi_i}$ gives
\begin{equation} \label{eq:boundSingularvalue}
    \| P\bm{\tilde{\alpha}}-\bm{1} \|_\pi \leq \frac{\max_i \sqrt{\pi_i}}{\min_i \sqrt{\pi_i}} \sigma_1 \|\bm{\tilde{\alpha}}-\bm{1} \|_\pi,
\end{equation}
and thus the bound of Eq. \ref{eq:bound_memoryvectorGENERAL} holds for $C = \sigma_1 \max_i \sqrt{\pi_i} / \min_i \sqrt{\pi_i}$. This bound is specially weak when $\min_i \sqrt{\pi_i} << \max_i \sqrt{\pi_i}$. 
\end{sloppypar}
\bigskip

\section{Rewriting the likelihood vectors} \label{sec:rewriting_likelihood}

In this section, we will rewrite the likelihood vectors $\rrr$ and $\aaa$ described in Section \ref{sec:partial_likelihoods}. We assume that all matrices $P_c$ have the same equilibrium distribution $\bm{\pi}$. Recall that the normalized likelihood vectors are defined as ${\rrt := \rrr/(\rrr \cdot \bm{\pi})}$ and ${\aat := \aaa/(\aaa \cdot \bm{\pi})}$. Assuming prior $\bm{\mu} = \bm{\pi}$, it holds that ${\text{Pr}(\partial) = \rrr \cdot \bm{\pi}}$ and ${\text{Pr}(\partial_c) = \aaa \cdot \bm{\pi}}$. Thus we define ${\text{Pr}_\pi(\partial) = \rrr \cdot \bm{\pi}}$ and ${\text{Pr}_\pi(\partial_c) = \aaa \cdot \bm{\pi}}$

\medskip

The normalized likelihood vectors are useful to identify the dependence between the subpatterns $\partial_1, \cdots, \partial_d$ composing pattern $\partial$. Indeed, since $\bm{\rho_\partial} = \bigcomp_{c \in [d]} P_c\aat$ and ${\text{Pr}_\pi(\partial) =  \rrr \cdot \bm{\pi}}$, it holds that

\begin{equation} \label{eq:likelihood_distribution_posterior} \text{Pr}_\pi(\partial) =  \text{Pr}_\pi(\partial_1, \cdots , \partial_d) =  \Big( \prod_{c \in [d]} \text{Pr}_\pi(\partial_c) \Big) \Big( \bm{\pi} \cdot \bigcomp_{c \in [d]} P_c \aat \Big),
\end{equation}
showing that the dependence factor between observations $\partial_1, \cdots , \partial_d$ is $\bm{\pi} \cdot \bigcomp_{c \in [d]} P_c \aat$. We additionally define the probability of observing $\partial$ assuming subpattern independence, \begin{equation}
    {\text{Pr}_{IND}(\partial) := \prod_{c\in [d]} \text{Pr}_\pi(\partial_c)}
\end{equation} and the dependence factor $D(\partial)$ of the subpatterns of $\partial$, \begin{equation}
    {D(\partial) := \bm{\pi} \cdot \bigcomp_{c \in [d]} P_c \aat},
\end{equation} yielding the relationship
\begin{equation} \label{eq:dependence_independence}
    \text{Pr}_\pi(\partial) = D(\partial) \text{Pr}_{IND}(\partial).
\end{equation}

\begin{sloppypar}
  By introducing $\aat := \aaa/ \text{Pr}_\pi(\partial_c)$ and the memory vectors $\bm{m_c} := P_c\aat - \bm{1}$ in equation $\bm{\rho_\partial} = \bigcomp_{c \in [d]} P_c\bm{\alpha^c_\partial}$, we obtain
\end{sloppypar}
 \begin{align} \label{eq:likelihood_expansion}
     \frac{\bm{\rho_\partial}}{\text{Pr}_{IND}(\partial)} &= 
     \bigcomp_{c\in [d]}\;  \Big( \bm{1} + \bm{m_c} \Big) = \nonumber\\
     &= \bm{1} + \sum_{p \in [d]} \sum_{C \in \binom{[d]}{p}} \bigcomp_{c \in C}   \bm{m_c}.
 \end{align}
 This equation combined with the relationship $\text{Pr}_\pi(\partial) = \bm{\rho_\partial} \cdot \bm{\pi} = D(\partial) \text{Pr}_{IND}(\partial)$ gives
 \begin{align} \label{eq:dependence_expansion}
      D(\partial) &=  \frac{\text{Pr}_\pi(\partial) }{\text{Pr}_{IND}(\partial)} = \frac{\bm{\rho_\partial}}{\text{Pr}_{IND}(\partial)}\cdot \bm{\pi} = \nonumber\\
      %&= \bm{1} \cdot \bm{\pi} + \sum_{c \in [d]} \bm{m_c} \cdot \bm{\pi}  + \sum_{p \in [2, d]} \sum_{C \in \binom{[d]}{p}}  \bm{\pi} \cdot \bigcomp_{c \in C}   \bm{m_c} = \nonumber \\
     &= 1 + \sum_{p \in [d]} \sum_{C \in \binom{[d]}{p}} \bm{\pi} \cdot \bigcomp_{c \in C}     \bm{m_c},
 \end{align}
 where the terms where $p = 1$ vanish, because $\bm{\pi}^T (P_c\aat - \bm{1}) = \bm{1} - \bm{1} = 0$. The expanded products of Equations \ref{eq:likelihood_expansion} and \ref{eq:dependence_expansion} will be useful to bound information flow in Section \ref{sec:try_again}.
 %where we used $\bm{m_c} \cdot \bm{\pi} = \bm{\pi}^T(P \bm{\tilde{\alpha}^c_\partial} - \bm{1}) = 1 - 1 = 0$.
      %&= 1 + \sum_{\{ c_1, c_2 \} \in  \binom{[d]}{2}} \sum_{ k\in [K]} (\bm{v_k} \cdot \bm{a^{c_1}_\partial})(\bm{v_k} \cdot \bm{a^{c_2}_\partial}) e^{2\lambda_kt} + \sum_{p \in [3, d]} \sum_{C \in \binom{[d]}{p}}  \bm{\pi} \cdot \bigcomp_{c \in C}   \bm{f^c_\partial}(t),  
\iffalse
\begin{align}
     |\bm{\pi} \cdot \bigcomp_{c \in C} \bm{m_c}| &=    |\bm{\pi}^{1 - p/2} \cdot \bigcomp_{c \in C} \bm{\pi}^{1/2} \circ \bm{m_c}| \leq \| \bm{\pi}^{1 - p/2} \| \prod_{c \in C}\|\bm{\pi}^{1/2} \circ \bm{m_c}| \| = \\
     &= \| \bm{\pi}^{1 - p/2} \| \prod_{c \in C} \| \bm{m_c} \|_\pi \leq \frac{1}{(\min \pi_i)^{p/2}} |\theta_1|^p \prod_{c \in C} \| \bm{\aat} - \bm{1}\|_\pi \leq\\
     &\leq \frac{1}{(\min \pi_i)^{p/2}} |\theta_1|^p (\frac{1}{\min_i \pi_i} - 1)^{p/2} = \psi^p,
\end{align}

where \[\psi := \frac{1}{(\min \pi_i)^{1/2}} |\theta_1| (\frac{1}{\min_i \pi_i} - 1)^{1/2} < \frac{|\theta_1|}{\min \pi_i}\].
\begin{align}
    |D(\partial) - 1| \leq \sum_{p \in [2,d]} \binom{d}{p} \psi^p < \frac{(d\psi)^2}{2} \frac{1}{1 - d \psi}
\end{align}

Assuming that $d\psi < 1/2$, it holds that 
\begin{equation}
     |D(\partial) - 1| < \epsilon \text{ if } d\psi < \sqrt{\epsilon}.
\end{equation}

\fi

\section{Bounds of information flow} \label{sec:try_again}

Along our proofs, we often need to bound the entrywise product by a sum of vectors. To this end, we prove the following.
\begin{lemma}\label{prop:poly_inequality}
For any integer $d \geq 2$ and reals $S\in (0, 2)$ and $\epsilon>0$, if $S \leq 4\epsilon/(1+2\epsilon)$, then 
\begin{equation} \label{eq:lemma}
    (1 + \frac{S}{d})^d < 1+ (1 + \epsilon) S.
\end{equation} 
\end{lemma}
\begin{proof}
Expanding $(1 + \frac{S}{d})^d$, Eq. \ref{eq:lemma} is equivalent to 
\begin{equation} \label{eq:dpolynomial2}
    \sum_{c \in [2,d]} \binom{d}{c} \frac{S^c}{d^c} < \epsilon S.
\end{equation}

Since $0 < S < 2$, we have the bound
\begin{equation} \label{eq:dpolynomial3}
    \sum_{c \in [2,d]} \binom{d}{c} \frac{S^c}{d^c} < \sum_{c \in [2,d]} \frac{S^c}{c!} < \sum_{c \in [2,\infty]} \frac{S^c}{2^c}  = \frac{S^2}{2(2 - S)}.
\end{equation}

It is easy to see that, for $S \in (0,2)$, inequality 
\begin{equation}
    \frac{S^2}{2(2 - S)} \leq \epsilon S
\end{equation}
holds iff 
\begin{equation}
    S \leq \frac{4\epsilon}{1 +2\epsilon},
\end{equation}
and the result follows.

\end{proof}

In Section \ref{sec:rewriting_likelihood}, we introduced the dependence factor $D(\partial)$, defined by Eq. \ref{eq:dependence_independence} as
\[\text{Pr}_\pi(\partial) =  \text{Pr}_\pi(\partial_1, \cdots , \partial_d) = D(\partial) \text{Pr}_{IND}(\partial).\]
Intuitively, if memory vectors $\bm{m_c}$ of Eq. \ref{eq:dependence_expansion} approach $\bm{0}$, then all patterns $\partial_1, \cdots, \partial_d$ are uninformative, and consequently they are also independent, that is, $D(\partial) \approx 1$. This is the content of the following lemma. Notably, in the statement of Lemma  \ref{prop:mixing_dependence}, normalized vectors $P_c\aat$ can be substituted by any normalized vectors $\bbt$. We use vectors $P_c\aat$ just to make the connection to other results of this work explicit.

\begin{lemma}[Mixing of the dependence factor]
\label{prop:mixing_dependence}
\begin{sloppypar}
 In the broadcasting process of Figure \ref{fig:multitree}, assume that all matrices $P_c$ have the same equilibrium frequency $\bm{\pi}$. Consider moreover the normalized likelihood vector $\rrt$ at $R$ and the normalized likelihood vectors  ${P_c\aat = \bm{1} + \bm{m_c}}$ at nodes $P_cA_c$. Then, under the assumption that  ${\sum_{c \in [d]}\| \bm{m_c}\|_\infty \leq 4\epsilon/(1 + 2\epsilon)}$, the dependence factor satisfies
\end{sloppypar}
\[|D(\partial) - 1|< \epsilon \sum_{c \in [d]}\|\bm{m_c} \|_\infty \leq \frac{4\epsilon^2}{1 + 2\epsilon}.\]

\end{lemma}
\begin{proof}
Let us bound the dependence factor $D(\partial)$ using the expansion of Eq. \ref{eq:dependence_expansion}. First, we denote by $|\bm{m_c}|$ the vector obtained by taking the absolute value of each entry of $\bm{m_c}$, and thus write $|\bm{m_c}| = (|m_c^0|, \cdots,  |m_c^K|).$ Using the fact that a weighted sum is smaller than its biggest term, plus the Arithmetic-Geometric Mean inequality, we obtain
\begin{align}
    |D(\partial) - 1| 
    &= |\bm{\pi} \cdot \sum_{p \in [2,d]} \sum_{C \in \binom{[d]}{p}}   \bigcomp_{c \in C}   \bm{m_c}|\nonumber\\
    &\leq \max_i \sum_{p \in [2,d]} \sum_{C \in \binom{[d]}{p}}   \prod_{c \in C}   |m_c^i| =\nonumber\\
    &= \max_i \prod_{i \in [c]} (1 + |m_c^i|) - \sum_{c\in [d]}  |m_c^i| - 1 \leq \nonumber\\
    &\leq \max_i  \; \; (1 + \frac{\sum_{c \in [d]}|m_c^i|}{d})^d - \sum_{c\in [d]}  |m_c^i| - 1. \label{eq:wherestopped}
\end{align}
Using Lemma \ref{prop:poly_inequality}, if $\sum_{c \in [d]}|m_c^i| < 4\epsilon/(1 + 2\epsilon)$ for all $i$, then the expression of Eq. \ref{eq:wherestopped} can be upperly, strictly bounded by $\epsilon \max_i \sum_{c \in [d]}|m_c^i|.$ Moreover, it is clear that
\[\sum_{c \in [d]}|m_c^i| \leq \sum_{c \in [d]}\|\bm{m_c} \|_\infty.\]
Therefore, assuming that the last sum is smaller than $4\epsilon/(1 + 2\epsilon)$, the dependence factor $D(\partial)$ satisfies
\begin{align}
    |D(\partial) - 1| 
    &< \epsilon \max_i  \sum_{c \in [d]}|m_c^i| \leq \nonumber\\ 
    &\leq \epsilon\sum_{c \in [d]}\|\bm{m_c} \|_\infty, \label{eq:wherecontinues}
\end{align}

proving the first inequality of the statement. The second inequality follows by reusing the assumption $\sum_{c \in [d]}\|\bm{m_c} \|_\infty \leq 4\epsilon/(1 + 2\epsilon)$.

\end{proof}

It is well known that, for $x\approx 0$, one can use the approximation ${(1 + x)^n \approx 1+ nx}$. In this sense, for numbers close to $1$, multiplication can be approximated by addition. In a similar fashion, if memory vectors $\bm{m_c} = P_c\aat - \bm{1}$ approach $\bm{0}$, then $P_c\aat \approx \bm{1}$ and the entrywise product $\bigcomp_{c \in [d]} P_c\aat$ can be approximated by $\bm{1} + \sum_{c \in [d]} \bm{m_c}$. A more subtle argument yields a linear bound of a memory vector given pattern $\partial$ using subpatterns $\partial_c$, as explained in Theorem \ref{prop:hadamard}. As in Lemma  \ref{prop:mixing_dependence}, we use vectors $P_c\aat$ just to make the connection to other results clear, and we could write any normalized vectors $\bbt$ instead.

\begin{theorem}[Hadamard-product upper bounds of the memory vector] \label{prop:hadamard}
\begin{sloppypar}
 In the broadcasting process of Figure \ref{fig:multitree}, assume that all matrices $P_c$ have the same equilibrium frequency $\bm{\pi}$. Consider moreover the normalized likelihood vector $\rrt$ at $R$ and the normalized likelihood vectors  $P_c\aat = \bm{1} + \bm{m_c}$ at nodes $P_cA_c$. Then, under the assumption that, for some $\epsilon>0$,  ${\sum_{c \in [d]} \|\bm{m_c} \|_\infty \leq 4\epsilon/(1 + 2\epsilon)}$, the following holds:
\end{sloppypar}
\begin{itemize}
    \item[a)] We have the bound \begin{equation} \label{eq:itema}
         D(\partial) \| \rrt - \bm{1}\|_\pi < (1+\epsilon)\sum _ {c \in [d]} \| \bm{m_c}\|_\pi.
    \end{equation}
    \item[b)] If moreover $\epsilon < (1 + \sqrt{5})/4$, then \[ \|\rrt - \bm{1}\|_\pi < \frac{1+\epsilon}{1 - \frac{4\epsilon^2}{1 + 2\epsilon}}\sum _ {c \in [d]} \| \bm{m_c}\|_\pi.
    \]
\end{itemize}
\end{theorem}
\begin{proof}
First of all, since in Eq. \ref{eq:dependence_independence} we stated that $\text{Pr}_\pi(\partial) = D(\partial) \text{Pr}_{IND}(\partial)$, it holds that
\begin{align}  \label{eq:memory_ind}
     D(\partial)\| \rrt - \bm{1}\|_\pi &= \| \frac{\rrr }{\text{Pr}_\pi(\partial)}D(\partial) - \bm{1}D(\partial)\|_\pi = \nonumber\\
   &=   \| \frac{\rrr}{\text{Pr}_{IND}(\partial)} - \bm{1} D(\partial)\|_\pi
\end{align}

Now we can use Equations \ref{eq:likelihood_expansion} and \ref{eq:dependence_expansion}, where we expanded $\rrr/ \text{Pr}_{IND}(\partial)$ and $D(\partial)$, giving
\begin{align} \label{eq:monster1}
     D(\partial)\| \rrt - \bm{1}\|_\pi &=  \| \bm{1} + \sum_{p \in [d]} \sum_{C \in \binom{[d]}{p}} \bigcomp_{c \in C}   \bm{m_c} - \bm{1} \Big( 1 + \sum_{p \in [d]} \sum_{C \in \binom{[d]}{p}}  \bm{\pi} \cdot \bigcomp_{c \in C}   \bm{m_c} \Big)\|_\pi =  \nonumber\\
     &= \| \sum_{p \in [d]} \sum_{C \in \binom{[d]}{p}}  \bigcomp_{c \in C}   \bm{m_c} - \bm{1} \Big(\bm{\pi} \cdot \sum_{p \in [d]} \sum_{C \in \binom{[d]}{p}} \bigcomp_{c \in C}   \bm{m_c} \Big) \|_\pi  \nonumber \\ 
     &=  \| \sum_{p \in [d]} \sum_{C \in \binom{[d]}{p}}\bigcomp_{c \in C}   \bm{m_c} - \bm{1} \langle \bm{1}, \sum_{p \in [d]} \sum_{C \in \binom{[d]}{p}} \bigcomp_{c \in C}   \bm{m_c}  \rangle_\pi \|_\pi.
\end{align}

Using the centralizing inequality of Prop. \ref{prop:inequality_L2}, we obtain
\begin{align}
     D(\partial)\| \rrt - \bm{1}\|_\pi &\leq    \| \sum_{p \in [d]} \sum_{C \in \binom{[d]}{p}}\bigcomp_{c \in C}   \bm{m_c}\|_\pi = \label{eq:monster10} \\
     &= \| \bigcomp_{c \in [d]} (\bm{1} + \bm{m_c}) - \bm{1}\|_\pi \label{eq:clean} = \\
     &= \| \; |\bigcomp_{c \in [d]} (\bm{1} + \bm{m_c}) - \bm{1} |\; \|_\pi \label{eq:cleanABSOLUTE},
\end{align}
where $|\cdot|$ denotes the entrywise absolute value.
 We have the entrywise inequality
 \begin{equation}
     |\bigcomp_{c \in [d]} (\bm{1} + \bm{m_c}) - \bm{1}| \leq \bigcomp_{c \in [d]} (\bm{1} + |\bm{m_c}|) - \bm{1},
 \end{equation}
 and the Arithmetic-Geometric Mean Inequality applied to $\bigcomp_{c \in [d]} (\bm{1} + |\bm{m_c}|)$ implies that, in an entrywise manner,

\[ 0 \leq  \bigcomp_{c \in [d]} (\bm{1} + |\bm{m_c}|) -\bm{1} \leq  \Big( \bm{1} + \frac{\sum_{c \in [d]} \bm{|m^c_\partial}|}{d}\Big)^d - \bm{1},\]
where the exponentiation occurs entrywise. Consequently 
\begin{align}
\label{eq:mon4}
     \| \bigcomp_{c \in [d]} (\bm{1} + \bm{m_c}) - \bm{1}\|_\pi &\leq   \| \Big( \bm{1} + \frac{\sum_{c \in [d]} \bm{|m^c_\partial}|}{d}\Big)^d - \bm{1}\|_\pi.
\end{align}
\begin{sloppypar}
  We write $|\bm{m_c}| = (|m_c^0|, \cdots,  |m_c^K|).$ Lemma \ref{prop:poly_inequality} with $S = \sum_{c \in [d]} |m_c^i|$ for $i \in \mathbb{A}$ further implies that, if we have ${\sum_{c \in [d]}\|\bm{m_c} \|_\infty \leq 4\epsilon/(1 + 2\epsilon)}$, then the following entrywise inequality holds
\end{sloppypar}
\begin{equation} \label{eq:Karound}
     \Big( \bm{1} + \frac{\sum_{c \in [d]} \bm{|m^c_\partial}|}{d}\Big)^d -\bm{1} < (1 + \epsilon)\sum_{c \in [d]} \bm{|m^c_\partial}|.
\end{equation}
Therefore, assuming that $\sum_{c \in [d]}\|\bm{m_c} \|_\infty \leq 4\epsilon/(1 + 2\epsilon)$, and combining Equations \ref{eq:cleanABSOLUTE}, \ref{eq:mon4} and \ref{eq:Karound}, it holds that \begin{align}
\label{eq:mon5}
     D(\partial)\| \rrt - \bm{1}\|_\pi &<  (1 + \epsilon) \| \sum_{c \in [d]} \bm{|m^c_\partial}|  \|_\pi \leq \nonumber\\
     &\leq  (1 + \epsilon)\sum_{c \in [d]} \|  \; |\bm{m_c}|  \; \|_\pi =  (1 + \epsilon) \sum_{c \in [d]} \|  \bm{m_c}  \|_\pi,
\end{align}
where we applied the triangle inequality. This proves item $a)$.

\medskip

\begin{sloppypar}
 To prove item $b)$, we use Prop. \ref{prop:mixing_dependence}, concretely the fact that, if ${\sum_{c \in [d]} \|\bm{m_c} \|_\infty \leq 4\epsilon/(1 + 2\epsilon)}$ and $1 - \frac{4\epsilon^2}{1 + 2\epsilon} > 0$, then
\end{sloppypar}
\begin{equation} \label{eq:dependence_bound}
    D(\partial) > 1 - \frac{4\epsilon^2}{1 + 2\epsilon} > 0.
\end{equation}
It is easy to see that $1 - \frac{4\epsilon^2}{1 + 2\epsilon}>0$ holds when $\epsilon > (1 + \sqrt{5})/4$. The inequality of Eq. \ref{eq:dependence_bound}, applied to Eq. \ref{eq:itema}, yields the result.

 \end{proof}
 
 Theorem \ref{prop:hadamard} does not consider the child nodes $A_c$, but only the parent nodes $P_cA_c$. Moreover, we did not mention how to make the memory vectors at $P_cA_c$ small enough to satisfy the assumption of Theorem \ref{prop:hadamard}. These particularities are considered in Theorem \ref{prop:rootByChildren}, which is our core result bounding the information flow on trees from children to parent node. Intuitively, one expects that information flows somehow decreasingly from the leaves towards the root. Theorem \ref{prop:rootByChildren} formalizes this intuition, sub-additively bounding the memory-vector norm at the root $R$ by the sum of norms at its child nodes $A_c$.

\begin{theorem}[Root-Children upper bounds of the memory vector] \label{prop:rootByChildren}
 In the broadcasting process of Figure \ref{fig:multitree}, assume that all matrices $P_c$ have the same equilibrium frequency $\bm{\pi}$. Consider moreover the normalized likelihood vector $\rrt$ at $R$ and the normalized likelihood vectors  $\aat$ at child nodes $A_c$. Moreover assume that, for some constants $C_c>0$ and for any normalized likelihood vector $\bm{\tilde{\alpha}}$, we have the bound \[ \|P_c\bm{\tilde{\alpha}}-\bm{1}\|_\pi \leq  C_c \; \|\bm{\tilde{\alpha}}-\bm{1}\|_\pi.\]
 
In these conditions, if  \[\sum_{c \in [d]} C_c  \leq \frac{\min_i \pi_i}{\sqrt{1-\min_i \pi_i}}  \frac{4\epsilon}{1 + 2\epsilon}\] for some $\epsilon>0$, then the following holds:
 \begin{itemize}
     \item[a)] We have the bound \[D(\partial) \| \rrt - \bm{1}\|_\pi < (1+\epsilon) \sum_{c \in [d]}C_c \|\aat - \bm{1} \|_\pi .\]
     \item[b)] If moreover $\epsilon < (1 +\sqrt{5})/4$, then \[\|\rrt - \bm{1}\|_\pi < \frac{1+\epsilon}{1 - \frac{4\epsilon^2}{1 + 2\epsilon}} \sum_{c \in [d]}C_c \|\aat - \bm{1} \|_\pi. \]
 \end{itemize}
 
\end{theorem}
 \begin{proof}
To apply  Theorem \ref{prop:hadamard}, we need the fulfillment of inequality
\[ {\sum_{c \in [d]} \|P_c\aat-\bm{1} \|_\infty \leq 4\epsilon/(1 + 2\epsilon)}.\]
Since $\|\bm{x} \|_\infty \leq \|\bm{x} \|_\pi/\min_i \sqrt{\pi_i}$ as stated in Eq. \ref{eq:chainNorms}, it holds that
\[\sum_{c \in [d]} \|P_c\aat-\bm{1} \|_\infty  \leq \frac{1}{\min_i \sqrt{\pi_i}} \sum_{c \in [d]} \|P_c\aat-\bm{1} \|_\pi.\]
We bounded the $L_2(\pi)$-norm of any memory vector in Eq. \ref{eq:memory}, implying
\begin{align}
    \sum_{c \in [d]}\|P_c\aat-\bm{1} \|_\pi &\leq \sum_{c \in [d]} C_c \|\aat-\bm{1} \|_\pi \leq \label{eq:sumBoundedC}\\   
    &\leq \sqrt{\frac{1}{\min_i \pi_i} -1} \sum_{c \in [d]} C_c\label{eq:sumBoundedNorm}.
\end{align}
Consequently, to apply Theorem \ref{prop:hadamard} it is enough that
\begin{equation*}
   \frac{1}{\min_i \sqrt{\pi_i}}\sqrt{\frac{1}{\min_i \pi_i} -1} \sum_{c \in [d]} C_c \leq  \frac{4\epsilon}{1 + 2\epsilon},
\end{equation*}
or equivalently

\begin{equation} \label{eq:quasifinal3}
   \sum_{c \in [d]} C_c  \leq \frac{\min_i \pi_i}{\sqrt{1-\min_i \pi_i}}  \frac{4\epsilon}{1 + 2\epsilon}.
\end{equation}

Equations \ref{eq:sumBoundedC} and \ref{eq:quasifinal3}, combined with Theorem \ref{prop:hadamard}, yield the desired results. 

 \end{proof}

\section{Definitions of unsolvability} \label{sec:definitions}
\subsection{Patternwise unsolvability}
In a broadcasting process on a $d$-ary tree, we want to describe when the root state cannot be confidently reconstructed using pattern $\partial$, assuming that the tree has a large number of levels $g$. In general, the channels of the broadcasting process can be different, although in Sections \ref{sec:bounds} and \ref{appn} we focus on the case when all channels are equal to $P$.

\begin{sloppypar}
 If we assume a prior state distribution $\bm{\mu} = (\mu_0, \cdots, \mu_K)$ and the likelihood vector at the root $R$ is $\rrr$, then Bayes' theorem implies that the posterior state distribution is \begin{equation}\bm{r_\partial} = \frac{\rrr \circ \bm{\mu}}{\text{Pr}(\partial)}   =\frac{\rrr \circ \bm{\mu} }{\rrr \cdot \bm{\mu}}.\end{equation} If we write $\bm{r_\partial} = (r_0, \cdots, r_K)$, then the MAP estimate is a state $i\in \mathbb{A}$ such that $r_i = \max_i r_i$, and the MAP reconstruction probability is $\max_i r_i$.
\end{sloppypar}

\begin{definition}
When $\bm{r_\partial} \neq \bm{\mu}$, we say that pattern $\partial$ is \textit{informative}, because the observation of $\partial$ influences the posterior. In contrast, we say that pattern $\partial$ is \textit{uninformative} when $\bm{r_\partial} = \bm{\mu}$.
\end{definition} 

Since ${\bm{r_\partial} =\rrr \circ \bm{\mu} / (\rrr \cdot \bm{\mu}),}$ pattern $\partial$ in uninformative iff $\rrr$ is uniform, which occurs iff $\rrt = \bm{1}$. The next definition states the kind of unsolvability that we will focus on.
\begin{definition} \label{def:unsolvable}
We say that the reconstruction problem is \textit{unsolvable} if every pattern $\partial$ is asymptotically uninformative as $g$ grows, that is, if
\[ \rrt \to \bm{1} \text{ as } g\to \infty.\]
\end{definition}
Since $\bm{\pi} > 0$, unsolvability is equivalent to the condition that, for all patterns $\partial$,
\[ \;\|\rrt - \bm{1}\|_\pi \to 0 \text{ as } g\to \infty. \]
More generally, since all finite-dimensional norms are equivalent (see \cite{MIT}), we can use any norm to measure $\rrt - \bm{1}$, although the $L_2(\pi)$-norm is convenient for our proofs. %One can even use F-norms (non homogeneous norms) such as $\|\rrt - \bm{1}\|^2_\pi$, which tends to zero iff $\|\rrt - \bm{1}\|_\pi$ tends to zero.

When the reconstruction problem is unsolvable, then, as $g \to \infty$, we have ${\bm{r_\partial} \to \bm{\mu}}$ and consequently ${\max_i r_i \to \max_i \mu_i}$, that is, the MAP reconstruction probability is achieved by just considering the prior. In other words, the observed pattern $\partial$ is uninformative about the ancestral root state of the broadcasting process. This is why we use the term "unsolvable", in the sense that we cannot estimate the ancestral root state better than the prior can do.

\subsection{Unsolvability in expectation}
%In Def. \ref{def:unsolvable}, we defined unsolvability using condition $\rrt \to \bm{1}$ as $g$ grows.
Notably, the set of all possible patterns $\Delta_g$ at the leaves of a $g$-level tree grows with $g$. Consequently, the expected value 
\begin{equation}
    \mathbb{E}_\partial[\;\|\rrt - \bm{1}\|_\pi]= \sum_{\partial \in \Delta_g} \Pr (\partial) \|\rrt - \bm{1}\|_\pi
\end{equation} has a growing number of summands as $g \to \infty$. In any case, 
the equivalence between finite-dimensional norms (see \cite{MIT}) implies that all definitions using norms are equivalent.

 \begin{definition} \label{def:unsolExpectation}
 We say that the reconstruction problem is  \textit{unsolvable in expectation} when, for some norm $\| \cdot \|_*$,
\[ \mathbb{E}_\partial[\;\|\rrt - \bm{1}\|_*] \to 0 \text{ as } g\to \infty. \]
 \end{definition}
 
 Unsolvability for $2$-dimensional or highly symmetric channels was studied in \cite{BeatingtheSecondEigenvalue} using the Total Variation (TV) distance.

\begin{definition}[Unsolvability in \cite{BeatingtheSecondEigenvalue}]\label{def:unsolTV}
 We say that the reconstruction problem is  TV-unsolvable if the likelihood vector ${\rrr = (\rho^0_\partial, \cdots, \rho^K_\partial)}$ satisfies, for all ${i \neq j \in \mathbb{A}}$, \[ \sum_{\partial \in \Delta_g} |\rho^i_\partial - \rho^j_\partial| \to 0 \text{ as } g\to \infty.\]
\end{definition}

Now define $\mathbb{E}^\pi_\partial[\| \rrt - 1\|_*]$ as the expected value of $\| \rrt - 1\|_*$ with prior $\bm{\mu} = \bm{\pi}$. Recall that $\bm{\mu} > 0$ by definition and $\bm{\pi} > 0$ due to Prop. \ref{prop:reversible_posi}.a. In Prop. \ref{prop:equivalentDEF}, we prove that Definitions \ref{def:unsolExpectation} and \ref{def:unsolTV} are equivalent.

\begin{prop} \label{prop:equivalentDEF}
The following statements are equivalent:

\begin{itemize}
    \item[a)] The reconstruction problem is unsolvable in expectation.
    \item[b)] The reconstruction problem is unsolvable in expectation assuming prior $\bm{\mu} = \bm{\pi}$.
    \item[c)] The reconstruction problem is TV-unsolvable.
\end{itemize}
\end{prop}
\begin{proof}
\hfill
\begin{itemize}
    \item $a) \iff b)$: Since $\text{Pr}_\pi (\partial) = \bm{\pi} \cdot \rrt$ and $\text{Pr} (\partial) = \bm{\mu} \cdot \rrt$, it holds that
\begin{equation}
     \text{Pr} (\partial) \min_i \pi_i  \leq 
     \text{Pr}_\pi (\partial) \leq 
     \frac{\text{Pr} (\partial)}{\min_i \mu_i},
\end{equation}
and consequently
\begin{equation}
    \mathbb{E}_\partial[\| \rrt - 1\|_*] \min_i \pi_i \leq  \mathbb{E}^\pi_\partial[\| \rrt - 1\|_*] \leq
    \frac{\mathbb{E}_\partial[\| \rrt - 1\|_*]}{\min_i \mu_i}.
\end{equation}
These chain of inequalities implies that $\mathbb{E}_\partial[\| \rrt - 1\|_*] \to 0$ iff $\mathbb{E}^\pi_\partial[\| \rrt - 1\|_*]\to 0$, as desired. Note that we can chose any prior as long as it has positive entries.

 \item $b) \iff c)$: If we write $\rrt = (\tilde{\rho}^0_\partial, \cdots, \tilde{\rho}^K_\partial)$, we can do 
\begin{align}
    \sum_{\partial \in \Delta_g} |\rho^i_\partial - \rho^j_\partial| =  \sum_{\partial \in \Delta_g} \text{Pr}_\pi(\partial) |\tilde{\rho}^i_\partial - \tilde{\rho}^j_\partial| = \mathbb{E}^\pi_\partial[|\tilde{\rho}^i_\partial - \tilde{\rho}^j_\partial|]
\end{align}

We have the inequality \begin{equation}
    |\tilde{\rho}^i_\partial - \tilde{\rho}^j_\partial| \leq |\tilde{\rho}^i_\partial-1| + |\tilde{\rho}^j_\partial-1|,
\end{equation}
\begin{sloppypar}
 and consequently unsolvability in expectation (using the $L_1$-norm) implies TV-unsolvability. Conversely, since $\bm{\pi}$ is a distribution and $\bm{\pi} \cdot \rrt = 1,$ it holds that
\end{sloppypar}
\begin{align}
    |\tilde{\rho}_\partial^i - 1| &= |\sum_j \tilde{\rho}_\partial^i \pi_j - \sum_j \tilde{\rho}_\partial^j \pi_j| = \nonumber\\
    &= |\sum_{j\neq i}  \pi_j(\tilde{\rho}_\partial^i - \tilde{\rho}_\partial^j)| \leq \sum_{j\neq i} \pi_j|(\tilde{\rho}_\partial^i - \tilde{\rho}_\partial^j)|.
\end{align}
Thus if $\mathbb{E}^\pi_\partial[|\tilde{\rho}^i_\partial - \tilde{\rho}^j_\partial|] \to 0$ for all $i \neq j$ then $\mathbb{E}_\partial[\| \rrt - 1\|_1] \to 0$, as desired. 
\end{itemize}

\end{proof}

\bigskip

\section{Bounds to the unsolvability problem} \label{sec:bounds}
Theorem \ref{prop:rootByChildren} is very general, because the only tree structure required is a parent node with $d$ children. By applying Theorem \ref{prop:rootByChildren} recursively, for any tree we can bound the memory vector at the root by the sum of the memory-vector norms at the leaves. Applying this strategy, we can find sufficient conditions for the unsolvability of the reconstruction problem.

\begin{theorem}\label{prop:unsolvable} Given any irreducible and aperiodic Markov matrix $P$, assume that, for some $C>0$ and for any normalized likelihood vector $\bm{\tilde{\alpha}}$, we have the bound \[ \|P\bm{\tilde{\alpha}}-\bm{1}\|_\pi \leq  C \; \|\bm{\tilde{\alpha}}-\bm{1}\|_\pi.\]
Then, on a $d$-ary tree, the reconstruction problem using channel $P$ is unsolvable if \[ C d < \min \Big\{ \frac{1}{3}, \frac{\min_i \pi_i}{\sqrt{1-\min_i \pi_i}} \Big\}.\]
\end{theorem}

\begin{proof}
 We want to apply Theorem \ref{prop:rootByChildren} with all constants $C_c$ equal to $C$. First note that, for any upper bound $U$, if the assumption $dC \leq U$  holds, then also $kC \leq U$ holds for any number of children $k \leq d$. Therefore, assuming $\epsilon < (1 +\sqrt{5})/4$ and \begin{equation} \label{eq:boundagain}
 d C  \leq \frac{\min_i \pi_i}{\sqrt{1-\min_i \pi_i}}  \frac{4\epsilon}{1 + 2\epsilon}\end{equation} as required by Theorem \ref{prop:rootByChildren}.b, then for any node with $k \leq d$ children we can apply the bound 
\begin{equation}\label{eq:boundkchildren}
\|\rrt - 1\|_\pi < \frac{1+\epsilon}{1 - \frac{4\epsilon^2}{1 + 2\epsilon}}C \sum_{c \in [k]}\|\aat - \bm{1} \|_\pi.\end{equation}
\begin{sloppypar}
 Consider the normalized likelihood vectors at the leaves, from now on $\llt$ where ${c \in \{ 1, \cdots, N\}}$ and $N \leq d^g$. The recursive application of Eq. \ref{eq:boundkchildren} yields that
\end{sloppypar}
\begin{align} \|\rrt - 1\|_\pi &< (\frac{1+\epsilon}{1 - \frac{4\epsilon^2}{1 + 2\epsilon}}C)^g \sum_{c \in [N]}\|\llt - \bm{1} \|_\pi,
\end{align}
while the upper bound of the norm of any memory vector (Eq. \ref{eq:memory}) gives
\begin{align}
\|\rrt - 1\|_\pi &< (\frac{1+\epsilon}{1 - \frac{4\epsilon^2}{1 + 2\epsilon}}C)^g \sum_{c \in [N]} \sqrt{\frac{1}{\min_i \pi_i} - 1} \leq
\nonumber\\ &\leq (\frac{1+\epsilon}{1 - \frac{4\epsilon^2}{1 + 2\epsilon}}C)^g d^g \sqrt{\frac{1}{\min_i \pi_i} - 1}. \label{eq:boundglevel}
\end{align}

Recall that the reconstruction problem is unsolvable iff $\|\rrt - 1\|_\pi$ tends to $0$ as $g$ grows. Thus the reconstruction problem is unsolvable if Eq. \ref{eq:boundagain} holds and
\begin{equation} \label{eq:boundglevel2}
    \frac{1+\epsilon}{1 - \frac{4\epsilon^2}{1 + 2\epsilon}}C d < 1,
\end{equation}
as implied by the bound of Eq. \ref{eq:boundglevel}.   To make the following bounds hold simultaneously,

\begin{align}
    C d &< \frac{1 - \frac{4\epsilon^2}{1 + 2\epsilon}}{1+\epsilon} \nonumber\\
    C d &\leq \frac{\min_i \pi_i}{\sqrt{1-\min_i \pi_i}}  \frac{4\epsilon}{1 + 2\epsilon},
\end{align}

the optimal $\epsilon< (1 +\sqrt{5})/4$ is the one making the two bounds coincide, although this intersection depends on $\min_i \pi_i$. For the sake of clarity, set $\epsilon = 1/2$, yielding the result.

\end{proof}

Theorem \ref{prop:unsolvable} assumes that any memory vector decreases at least by a factor of $C$ under the action of $P$. We have inferred such bounds in Equations \ref{eq:bound_memoryvector} and \ref{eq:boundSingularvalue}, giving the following corollaries.

\begin{cor}
Consider any irreducible and aperiodic Markov matrix $P$ with largest singular value $\sigma_1$. Then, on a $d$-ary tree, the reconstruction problem using channel $P$ is unsolvable if \[ \sigma_1 d < \frac{\min_i \sqrt{\pi_i}}{\max_i \sqrt{\pi_i}} \min \Big\{ \frac{1}{3}, \frac{\min_i \pi_i}{\sqrt{1-\min_i \pi_i}} \Big\}.\]
\end{cor}

\begin{cor} \label{cor:boundREVERSIBLE}
Consider any irreducible, aperiodic and reversible Markov matrix $P$ with non-unitary absolutely largest eigenvalue $\theta_1$. Then, on a $d$-ary tree, the reconstruction problem using channel $P$ is unsolvable if \[ |\theta_1| d < \min \Big\{ \frac{1}{3}, \frac{\min_i \pi_i}{\sqrt{1-\min_i \pi_i}} \Big\}.\]
\end{cor}

%The choice of $\epsilon$ can be made more appropriate if we know more about $d$ and ${\min_i \pi_i \leq 1/|\mathbb{A}|}$. For example, we can choose $\epsilon = 2/3$, and the sufficient condition would be \[|\theta_1| d < \min \Big\{ \frac{1}{7}, \frac{8}{7} \frac{\min_i \pi_i}{ \sqrt{1-\min_i \pi_i}} \Big\}.\]
In this section we only assumed that the tree is $d$-ary, but Theorem \ref{prop:rootByChildren} is very general and can be applied in multiple scenarios. We can, for example, allow a different channel for every child node as long as $\bm{\pi}$ is the equilibrium distribution. To get a tighter bound, one can also apply Theorem \ref{prop:rootByChildren} using the exact number of children at every node, and not an upper bound $d$ as we did here.

\section{A bound of unsolvability in expectation}\label{appn} %% if no title is needed, leave empty \section*{}.
Clearly, if the reconstruction problem is unsolvable, then it is unsolvable in expectation. Moreover, there is room to slightly improve our bounds when dealing with unsolvability in expectation. First we can state the analogous of Theorem \ref{prop:rootByChildren}.
\begin{prop} \label{prop:rootByChildrenEXPECTATION}
 In the broadcasting process of Figure \ref{fig:multitree}, , assume that all matrices $P_c$ have the same equilibrium frequency $\bm{\pi}$. We consider for each pattern $\partial$ the normalized likelihood vector $\rrt$ at $R$ and the normalized likelihood vectors $\aap$ at child nodes $A_c$. Moreover assume that, for some constants $C_c>0$ and for any normalized likelihood vector $\bm{\tilde{\alpha}}$, we have the bound \[ \|P_c\bm{\tilde{\alpha}}-\bm{1}\|_\pi \leq  C_c \; \|\bm{\tilde{\alpha}}-\bm{1}\|_\pi.\]
 
In these conditions, if  \[\sum_{c \in [d]} C_c  \leq \frac{\min_i \pi_i}{\sqrt{1-\min_i \pi_i}}  \frac{4\epsilon}{1 + 2\epsilon}\] for some $\epsilon>0$, then we have the bound \[\mathbb{E}^\pi_\partial[\| \rrt - 1\|_\pi] < (1+\epsilon) \sum_{c \in [d]}C_c \mathbb{E}^\pi_{\partial_c}[\|\aap - \bm{1} \|_\pi].\]
\end{prop}

\begin{proof}
We can apply Theorem \ref{prop:rootByChildren}.a, because the assumptions are identical. Multiplying both sides by $\text{Pr}_{IND}(\partial)$, we get
\begin{equation}
    \text{Pr}_\pi(\partial)\| \rrt - \bm{1}\|_\pi < (1+\epsilon) \sum_{c \in [d]}C_c \text{Pr}_{IND}(\partial)\|\aap - \bm{1} \|_\pi
\end{equation}
Recall that pattern $\partial$ is composed by subpatterns $\partial_1, \cdots, \partial_d$, and thus indexation using $\partial$ is equivalent to indexation using $\partial_1, \cdots, \partial_d$. Moreover, vector $\aap$ depends only on $\partial_c$, so we can write $\aap = \aapp$. Summing over all $\partial$, it follows that
\begin{align}
    \mathbb{E}^\pi_\partial[\| \rrt - \bm{1}\|_\pi &< (1+\epsilon) \sum_{c \in [d]} C_c \sum_{\partial_1, \cdots, \partial_d} \text{Pr}_{IND}(\partial)\|\aapp - \bm{1}\|_\pi = \nonumber\\
    &= (1+\epsilon) \sum_{c \in [d]} C_c  \sum_{\partial_c}\text{Pr}_\pi(\partial_c)\|\aapp - \bm{1}\|_\pi  \prod_{f\in [d]-\{c\}} \sum_{\partial_f} \text{Pr}_\pi(\partial_f),
\end{align}
where we substituted $\text{Pr}_{IND}(\partial) = \prod_{f\in [d]}\text{Pr}_\pi(\partial_f)$ and rearranged the summands. Since ${\sum_{\partial_f} \text{Pr}_\pi(\partial_f) = 1}$ for all $f \in [d]$, we get the desired inequality
\begin{align}
    \mathbb{E}^\pi_\partial[\| \rrt - \bm{1}\|_\pi] &< (1+\epsilon) \sum_{c \in [d]} C_c  \sum_{\partial_c}\text{Pr}_\pi(\partial_c)\|\aapp - \bm{1}\|_\pi \nonumber\\ &= (1+\epsilon) \sum_{c \in [d]} C_c  \mathbb{E}_{\partial_c}[\|\aapp - \bm{1}\|_\pi].
\end{align}
\end{proof}

Using Proposition \ref{prop:rootByChildrenEXPECTATION}, we can improve the unsolvability bounds of Section \ref{sec:bounds}.
\begin{prop}
Given any irreducible and aperiodic Markov matrix $P$, assume that, for some $C>0$ and for any normalized likelihood vector $\bm{\tilde{\alpha}}$, we have the bound \[ \|P\bm{\tilde{\alpha}}-\bm{1}\|_\pi \leq  C \; \|\bm{\tilde{\alpha}}-\bm{1}\|_\pi.\]
Then, on a $d$-ary tree, the reconstruction problem using channel $P$ is unsolvable in expectation if \[ C d < \min \Big\{ \frac{1}{3}, \frac{8}{5}\frac{\min_i \pi_i}{\sqrt{1-\min_i \pi_i}} \Big\}.\]
\end{prop}
\begin{proof}
Applying Prop. \ref{prop:rootByChildrenEXPECTATION}, the proof is analogous to that of Theorem \ref{prop:unsolvable}. Eventually we need to make the following two bounds hold simultaneously,
\begin{align}
    C d &< \frac{1}{1+\epsilon} \nonumber\\
    C d &\leq \frac{\min_i \pi_i}{\sqrt{1-\min_i \pi_i}}  \frac{4\epsilon}{1 + 2\epsilon}.    
\end{align}
Picking $\epsilon = 2$ yields the result.

\end{proof}

\iffalse

%%%%%%%%%%%%%%%%%%%%%%%%%%%%%%%%%%%%%%%%%%%%%%
%% Example with multiple Appendixes:        %%
%%%%%%%%%%%%%%%%%%%%%%%%%%%%%%%%%%%%%%%%%%%%%%
\begin{appendix}
\section{Title of the first appendix}\label{appA}
If there are more than one appendix, then please refer to it
as \ldots\ in Appendix \ref{appA}, Appendix \ref{appB}, etc.

\section{Title of the second appendix}\label{appB}
\subsection{First subsection of Appendix \protect\ref{appB}}

Use the standard \LaTeX\ commands for headings in \verb|{appendix}|.
Headings and other objects will be numbered automatically.
\begin{equation}
\mathcal{P}=(j_{k,1},j_{k,2},\dots,j_{k,m(k)}). \label{path}
\end{equation}

Sample of cross-reference to the formula (\ref{path}) in Appendix \ref{appB}.
\end{appendix}
\fi
%%%%%%%%%%%%%%%%%%%%%%%%%%%%%%%%%%%%%%%%%%%%%%
%% Support information, if any,             %%
%% should be provided in the                %%
%% Acknowledgements section.                %%
%%%%%%%%%%%%%%%%%%%%%%%%%%%%%%%%%%%%%%%%%%%%%%
\begin{acks}[Acknowledgments]
The author would like to thank Arndt von Haeseler for his useful comments that improved this paper.
\end{acks}

%%%%%%%%%%%%%%%%%%%%%%%%%%%%%%%%%%%%%%%%%%%%%%
%% Funding information, if any,             %%
%% should be provided in the                %%
%% funding section.                         %%
%%%%%%%%%%%%%%%%%%%%%%%%%%%%%%%%%%%%%%%%%%%%%%
\begin{funding}
C.M. was partly supported by the Austrian Science Fund (FWF, grant number I-1824-B22) to Arndt von Haeseler.
\end{funding}

%%  MR numbers will be added by VTeX.                      %%
%%                                                         %%
%%  Use \cite{...} to cite references in text.             %%
%%                                                         %%
%%%%%%%%%%%%%%%%%%%%%%%%%%%%%%%%%%%%%%%%%%%%%%%%%%%%%%%%%%%%%

%% if your bibliography is in bibtex format, uncomment commands:
\bibliographystyle{imsart-number} % Style BST file (imsart-number.bst or imsart-nameyear.bst)
\bibliography{aap-template}       % Bibliography file (usually '*.bib')

\end{document}